\documentclass[journal]{IEEEtran}

\usepackage[cmex10]{amsmath}
\usepackage{amssymb,bm,amsfonts,mathrsfs}
\usepackage[pdftex]{graphicx}
\usepackage{pgfplots}
\pgfplotsset{compat=1.17}
\usepackage{algorithm}
\usepackage{algpseudocode}
\usepackage{array,booktabs,tabularx,multirow,makecell}
\usepackage{subcaption}
\usepackage{stfloats}
\usepackage{siunitx}
\usepackage{gensymb}
\usepackage[draft,bookmarks=false]{hyperref}
\usepackage{tikz}
\usetikzlibrary{decorations.markings,decorations.pathmorphing}
\usepackage[noadjust]{cite}
\usepackage{xcolor}
\usepackage{url} 
\allowdisplaybreaks

\newtheorem{proposition}{Proposition}
\newtheorem{remark}{Remark}
\def\BibTeX{{\rm B\kern-.05em{\sc i\kern-.025em b}\kern-.08em
    T\kern-.1667em\lower.7ex\hbox{E}\kern-.125emX}}

\begin{document}
\title{Enabling mmWave Communications with VCSEL-Based Light-Emitting Reconfigurable Intelligent Surfaces}
\author{
Rashid Iqbal\textsuperscript{*},
Dimitrios Bozanis\textsuperscript{\dag},
Dimitrios Tyrovolas\textsuperscript{\dag},
Christos K. Liaskos\textsuperscript{\ddag},\\
Muhammad Ali Imran\textsuperscript{*},
George K. Karagiannidis\textsuperscript{\dag},
Hanaa Abumarshoud\textsuperscript{*}\\[2pt]
\textsuperscript{*}James Watt School of Engineering, University of Glasgow, Glasgow G12 8QQ, U.K.\\
e-mail: r.iqbal.1@research.gla.ac.uk, \{muhammad.imran, hanaa.abumarshoud\}@glasgow.ac.uk\\
\textsuperscript{\dag}Department of Electrical and Computer Engineering, Aristotle University of Thessaloniki, 54124 Thessaloniki, Greece\\
e-mail: \{dimimpoz, tyrovolas, geokarag\}@auth.gr\\
\textsuperscript{\ddag}University of Ioannina and FORTH, Greece\\
e-mail: cliaskos@uoi.gr
}
\maketitle
\thispagestyle{empty}
\pagestyle{empty}
\begin{abstract}
This paper proposes a light-emitting reconfigurable intelligent surface (LeRIS) architecture that integrates vertical cavity surface-emitting lasers (VCSELs) to jointly support user localization and mmWave communication. By leveraging the directional Gaussian beams and dual-mode diversity of VCSELs, we derive a closed-form method for estimating user position and orientation using only three VCSEL sources. These estimates are then used to configure LeRIS panels for directional mmWave beamforming, enabling optimized wave propagation in programmable wireless environments. Simulation results demonstrate that the proposed system achieves millimeter-level localization accuracy and maintains high spectral efficiency. These findings establish VCSEL-integrated LeRIS as a scalable and multifunctional solution for future 6G programmable wireless environments.
\end{abstract}
\vspace{-2mm}

\IEEEpeerreviewmaketitle

\vspace{-5mm}

\section{Introduction}
The evolution toward sixth generation networks is fueled by diverse services and applications such as holographic telepresence, the metaverse, and industrial automation, which demand massive user access, high bandwidth efficiency, and extremely low latency \cite{6GApplications}. To satisfy these requirements, future systems will increasingly rely on millimeter wave (mmWave) frequencies, yet operation in such bands faces critical challenges, including severe path loss and susceptibility to blockages. To address these obstacles, the concept of programmable wireless environments (PWEs) has been introduced, where reconfigurable intelligent surfaces (RISs) reshape wave propagation on demand and provide the precise manipulation of electromagnetic signals required for reliable mmWave operation. The effectiveness of this paradigm, however, depends on accurate awareness of user positions, since inaccurate knowledge can lead to suboptimal RIS configuration and degraded service quality. This realization has motivated the development of multi-functional RISs, which extend PWEs beyond communication enhancement by incorporating localization capabilities \cite{MFRIS2025}. Current research primarily explores this through enhanced electromagnetic capabilities using simultaneous reflection and refraction combined with advanced configuration schemes \cite{MFRIS2025}. However, relying solely on electromagnetic configuration may introduce performance trade-offs, since optimizing the RIS for one task can compromise its effectiveness for another. As a result, multi-functional RISs are expected to evolve through lightweight architectural extensions, where complementary elements enhance environmental awareness without undermining their fundamental wavefront manipulation role.

Among possible extensions, optical sources stand out as candidates since their predominantly line-of-sight behavior provides stable references for position and orientation through multiple inherent degrees of freedom, while their operation in a distinct spectral domain ensures isolation from the spectrum manipulated by RISs \cite{Tyrovolas2025Leris}. Several works have investigated light-emitting diode (LED) systems as optical anchors for localization, achieving centimeter-level positioning accuracy \cite{boz_loc22, boz_loc25}. However, their diffuse emission patterns constrain spatial resolution and preclude precise orientation estimation. Building on these observations, vertical cavity surface-emitting lasers (VCSELs) have emerged as compact alternatives that combine precise beam control with fast modulation and low-power operation. Their highly directional Gaussian-profile beams provide increased angular resolution, while multiple optical modes offer additional diversity for enhanced localization. In this direction, the authors of \cite{Zeng2021_VCSEL_BeamActivation} proposed a VCSEL array system with novel beam activation schemes that maintained connectivity under random user orientation while analyzing the tradeoff between divergence and throughput, while the authors of \cite{R1} proposed VCSEL-based approaches using deep neural networks for joint position and orientation estimation with lower localization error compared to LED-based systems, and the authors of \cite{sarbazi} investigated ultra-dense VCSEL array architectures supporting multi-gigabit transmission.

The concept of light-emitting RISs (LeRIS) has been introduced, showing that embedding optical anchors within RIS panels can enhance localization accuracy and directly assist RIS configuration for improved spectral efficiency \cite{Tyrovolas2025Leris}. However, current LeRIS designs have been limited to LED implementations, and to the best of the authors' knowledge, no prior work has established closed-form localization and orientation estimation schemes that exploit the dual-mode spatial degrees of freedom of VCSELs for LeRIS.

Motivated by the above, this paper introduces a VCSEL-based LeRIS-assisted PWE architecture that jointly supports localization, orientation estimation, and mmWave communication. Specifically, we introduce an architecture where three dual-mode VCSEL sources are integrated along the perimeter of the panel to act as structured optical anchors. We then derive closed-form expressions that exploit the Gaussian beam profile and dual-mode diversity to ensure joint position and orientation estimation. Through extensive simulations, we demonstrate that the proposed system achieves millimeter-level localization accuracy and robust orientation estimation while sustaining substantial spectral efficiency gains for mmWave communication, thereby establishing VCSEL-based LeRIS as a practical pathway for extending RIS capabilities toward multi-functional operation within PWEs.
\vspace{-2.5mm}
\section{System Model}\label{sec:II}\vspace{-1mm}    
We consider a PWE in an indoor space, Fig.~\ref{fig:system_model}, where a mmWave access point (AP) with a directional antenna is outside the room and serves a single user equipment (UE) inside. The UE has a directional antenna and an optical receiver. As mmWave signals are susceptible to blockage in indoor spaces, a direct LoS link from the AP to the UE is often infeasible. To address this, four LeRIS panels are mounted at the centers of the walls \cite{R2}. Each LeRIS integrates a programmable array of passive reflecting elements and VCSELs, enabling localization and joint control of electromagnetic reflection for dynamic routing based on UE awareness.

To localize the user, the UE estimates its position from optical signals emitted by VCSELs around the LeRIS perimeters. The VCSELs produce narrow Gaussian beams that scan small angular ranges to cover azimuth sectors, which are more appropriate for localization services due to their directivity, compared to the diffusive nature of LEDs. Each VCSEL uses a distinct infrared frequency, allowing the UE to separate the signals spectrally, infer its position, and report it to the AP over a control channel. The AP then configures the LeRIS panels to form directional mmWave links to the UE.
\vspace{-2mm}    
\begin{figure}[h]
\centering
\includegraphics[width=0.8\linewidth]{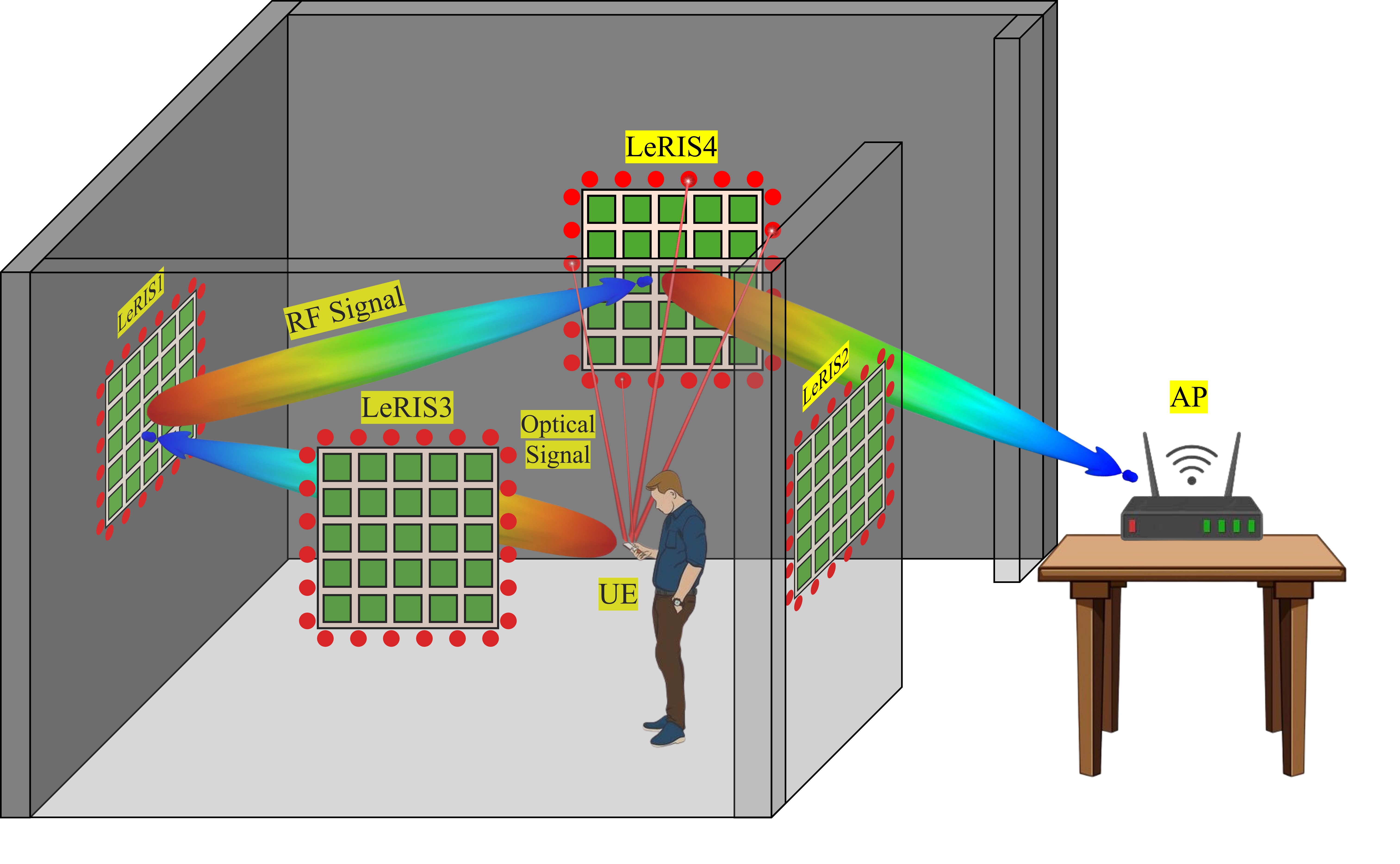}
\captionsetup{font=small}
\caption{Programmable Wireless environment with VCSEL-based LeRIS panels.}
\label{fig:system_model}
\vspace{-6mm}    
\end{figure}
  
\subsection{VCSELs Channel Modeling}\vspace{-1mm}   
To support high-throughput communication under dynamic indoor conditions, the PWE leverages the deployed LeRISs to form a mmWave signal route toward each selected UE, by utilizing the transmitted optical signals of the VCSELs to extract its location. Unlike conventional diffuse light emitters, such as LEDs, which exhibit broad angular dispersion and lower spatial coherence, VCSELs generate narrow, low-divergence beams with Gaussian spatial intensity profiles that can be tightly controlled and directionally steered, thus being particularly suitable for precise localization.
The optical channel between the $i$-th VCSEL and the UE is characterized by the intensity $I_{i}$ received at the user, assumed to follow a Gaussian distribution~\cite{R1}:
{\small
\begin{equation}
I_{i}(r_0, d_0) = \frac{2P_{t,i}}{\pi w^2(d_0)} \exp\left(-\frac{2r_0^2}{w^2(d_0)}\right),
\label{eq:gaussian_intensity}
\end{equation}}
\vspace{-2.5mm}\\
where $P_{t,i}$ is the transmitted optical power, $r_0$ is the radial distance from the beam center, and $d_0$ is the axial distance. The beam spot size $w(d_0)$ evolves as
{\small
\begin{equation}
w(d_0)=w_0 \sqrt{1+\left(\frac{d_0}{z_{\mathrm R}}\right)^{2}},
\label{eq:beam_spot}
\end{equation}}
\vspace{-2.5mm}\\
where $w_0 = \frac{\lambda_o}{\pi \theta_{\mathrm{div}}}$ is the beam waist, $\theta_{\mathrm{div}}$ is the divergence angle, $\lambda_o$ is the wavelength, and $z_{\mathrm R} = \frac{\pi w_0^{2}}{\lambda_o}$ is the Rayleigh range.
The angular intensity at the UE becomes
{\small
\begin{equation}
I_{i}(d_i, \phi_i) = \frac{2P_{t,i}}{\pi w^2(d_{i} \cos \phi_i)} \exp\left(-\frac{2d_{i}^2 \sin^2 \phi_i}{w^2(d_{i} \cos \phi_i)}\right),
\label{eq:angled_intensity}
\end{equation}}
\vspace{-2.5mm}\\
where $d_i$ is the distance and $\phi_i$ the irradiance angle. The received LoS optical power at the UE’s photodetector (PD) is then
\begin{equation}
P_{\mathrm{LoS},i} = I_{i}(d_i, \phi_i) A_{\mathrm{PD}} \cos \psi_{i} \, \mathrm{rect}\left(\frac{\psi_{i}}{\Psi}\right),
\label{eq:received_power}
\end{equation}
with $A_{\mathrm{PD}}$ denoting PD area, $\psi_i$ the incidence angle, and $\Psi$ the half-angle field of view (FoV). The $\mathrm{rect}(\cdot)$ function enforces angular filtering based on the FoV. However, the total power measured at the PD includes noise, yielding
\begin{equation} \label{Noise}
P_{r,i} = P_{\mathrm{LoS},i}+ P_{n,i},
\end{equation}
where $P_{n,i} = B_o S_i$ models the noise variance, with $S_i$ the one-sided power spectral density
\begin{equation}
\begin{split}
S_{i} &= A_{K}
+ R_\mathrm{PD}  P_{\mathrm{LoS},i}\left( 2q  + \mathrm{RIN} R_\mathrm{PD}  P_{\mathrm{LoS},i}\right),
\end{split}
\label{eq:noise_psd}
\end{equation}
and $B_o$ the PD bandwidth. Here, $A_K = \frac{4 k_{B} T F_n}{R_L}$ is the thermal noise term, where $k_{B}$ is Boltzmann’s constant, $T$ is temperature, $R_L$ is load resistance, and $F_n$ the preamplifier noise. The terms $q$, $R_{\mathrm{PD}}$, and $\mathrm{RIN}$ denote the electron charge, PD responsivity, and relative intensity noise, respectively \cite{Kazemi}. Thus, the optical channel captures the physical beam propagation, receiver orientation, and power conversion effects.  
\subsection{mmWave Channel Modelling}     
Based on the estimated user location and orientation, the AP selects a cascaded route $L$ of participating LeRISs that connects the AP to the UE, where each LeRIS along this path is configured to steer the incident mmWave signal toward the next LeRIS, while the final LeRIS steers energy toward the UE according to its estimated angular location $(\hat{\theta}_{r}, \hat{\phi}_{r})$. For this purpose, each LeRIS panel, composed of $M \times N$ square reflecting elements with side length $D$, shapes the outgoing wavefront based on the intended direction. Thus, the effective gain of each LeRIS produced in the target direction $(\theta, \phi)$ is given by 
{\small
\begin{equation}
    G \left(\theta, \phi \right) = {\eta_{\mathrm{eff}}}  \dfrac{4 \pi \lvert F\left(\theta, \phi \right)\rvert ^2}{\int_{0}^{2\pi} \int_{0}^{\frac{\pi}{2}}\lvert F\left(\theta, \phi \right)\rvert ^2 \mathrm{sin}\left(\theta\right) d\theta d\phi},
    \end{equation}}\\
where $\eta_{\mathrm{eff}}$ denotes the aperture efficiency~\cite{cui}. Moreover, $F(\theta, \phi)$ is the far-field radiation pattern generated by the LeRIS array, which is given as \cite{R2}
{\small
\begin{equation}
F(\theta, \phi) = \sum_{m=1}^M \sum_{n=1}^N e^{j\left(k_0 \zeta_{mn}(\theta,\phi) + \omega_{mn} + \Phi_{mn}\right)},
\end{equation}}\\
where $k_0 = \frac{2\pi}{\lambda_m}$ is the wavenumber and $\lambda_m$ is the wavelength of the mmWave channel, and the term $\zeta_{mn}(\theta, \phi)$ expresses the angular-dependent phase shift associated with the array’s geometry and can be expressed as \cite{R2}
{\small
\begin{equation}
\begin{split}
\zeta_{mn}(\theta, \phi) = D \sin(\theta) \left[ \left(m - \tfrac{1}{2} \right) \cos(\phi) + \left(n - \tfrac{1}{2} \right) \sin(\phi) \right] \\
+ (x_T - x_{R}) \sin(\theta) \cos(\phi)
+ (y_T - y_{R}) \sin(\theta) \sin(\phi) \\
+ (z_T - z_{R}) \cos(\theta),
\end{split}
\end{equation}}\\
where $(x_T, y_T, z_T)$ are the coordinates of the transmitting or reflecting node, and $(x_{R}, y_{R}, z_{R})$ denote the center of the receiving node. Additionally, the phase offset $\omega_{mn}$, corresponding to the path delay from transmitter to the $(m,n)$-th reflecting element, is given by \cite{R2}
{\small
\begin{equation}
\begin{split}
&\omega_{mn} = k_0 \Big( \left(x_T - D\left(m - \tfrac{1}{2}\right) \sin(\theta)\cos(\phi) - x_{R} \right)^2 \\
&+ \left(y_T - y_{R}\right)^2 \\& + \left(z_T - D\left(m - \tfrac{1}{2}\right) \sin(\theta)\sin(\phi) - z_{R} \right)^2 \Big)^{1/2}.
\end{split}
\end{equation}}\\
Finally, to steer energy accurately in the direction $(\hat{\theta}_{r}, \hat{\phi}_{r})$, the phase shift $\Phi_{mn}$ applied by each element is configured as~\cite{Abadal}
{\small
\begin{equation}
\begin{split}
\Phi_{mn} &= -k_0 D \left[m \cos(\hat{\phi}_{r}) \sin(\hat{\theta}_{r}) + n \sin(\hat{\phi}_{r}) \sin(\hat{\theta}_{r}) \right] \\& - \omega_{mn}.
\end{split}
\end{equation}}\\
When steering toward another LeRIS panel, the exact position and orientation of the target node are known. Under such conditions, the beam is perfectly aligned with the desired direction, thus, the far-field response achieves maximum constructive interference, and the achievable gain becomes
{\small
\begin{equation}
    G_{\mathrm{max}}= \eta_{\mathrm{eff}} MN.
\end{equation}}\\
Therefore, given that each LeRIS along the route, excluding the last, performs perfect beam steering toward a known position, the cumulative gain contributed by the first $L-1$ LeRIS panels can be expressed as
{\small
\begin{equation}
G_{\mathrm{cas}} = \left(A_{\mathrm{eff}} G_{\mathrm{max}}\right)^{L-1},
\end{equation}}\\
where
{\small
\begin{equation}
A_{\mathrm{eff}} = \frac{MN \lambda_m^2}{4\pi},
\end{equation}}\\
denotes the effective aperture of each LeRIS panel~\cite{Abadal}. However, for the last participating LeRIS in the route, the signal will be steered toward the estimated UE direction $(\hat{\theta}_{r}, \hat{\phi}_{r})$, derived from the proposed optical localization system. Specifically, due to potential angular mismatch between the estimated and true UE direction $(\theta_{u}, \phi_{u})$, the effective beam alignment may degrade, resulting in a reduced directional gain~\cite{R2}. As a result, the achievable gain by the final LeRIS in the true UE direction is expressed as

\begin{equation}
G_L(\theta_{u}, \phi_{u}) = \eta_{\mathrm{eff}} \cdot \frac{4\pi \left|F(\theta_{u}, \phi_{u})\right|^2}{\int_0^{2\pi} \int_0^{\pi} \left|F(\theta, \phi)\right|^2 \sin \left(\theta\right) d\theta d\phi}.
\end{equation}
The UE is assumed to be equipped with a directional antenna featuring a constrained FoV, which limits the angular range over which the final LeRIS-to-UE link can be established. Specifically, a viable connection is possible only when the incoming beam direction lies within the overlapping FoVs of both the LeRIS and the UE. The corresponding directional gain $G_r(\theta_u)$ for a beam with directivity angle $\theta_m$ is
{\small
\begin{equation}
G_r(\theta_u)=
\begin{cases}
\frac{2\pi}{\theta_m}, & \text{if } |\theta_u|\le \theta_m, \\\\
0, & \text{otherwise}.
\end{cases}
\label{eq:user_gain}
\end{equation}}
\vspace{-2.5mm}\\
Consequently, the total directional gain accumulated along the entire route is expressed as
{\small
\begin{equation}
G_{\mathrm{r}} = G_{\mathrm{cas}} \cdot A_{\mathrm{eff}} \cdot G_L(\theta_{u}, \phi_{u}) \cdot G_r(\theta_u),
\end{equation}}
\vspace{-2.5mm}\\
where the final factor accounts for the UE’s directional reception. Assuming narrowband flat-fading mmWave transmission, the baseband signal received by the UE is given by
\begin{equation}
y = \sqrt{l_p G_t P_t G_{\mathrm{r}}} \cdot x + w,
\end{equation}
where $x$ is the unit-energy transmitted symbol, $G_t$ is the AP antenna gain, $P_t$ is the transmit power, and $w$ denotes the AWGN at the UE with variance $\sigma^2$. Moreover, assuming a cascaded LeRIS network with $L$ LeRISs, $l_p$ denotes the total path loss over the route, given by~\cite{Cascaded}
{\small
\begin{equation}
l_p = \prod_{i=1}^{L+1} C_0 \left(\frac{d_{r,i}}{d_{r,0}}\right)^{-n_i},
\end{equation}
}
\vspace{-2.5mm}\\
where $d_{r,i}$ represents the distance of each segment, $n_i$ is the path loss exponent for the $i$-th segment, and $C_0 = \frac{\lambda_m^2}{(4\pi d_{r,0})^2}$ is the free space path loss at the reference distance $d_{r,0}$. To this end, the achievable spectral efficiency at the user is then given by
{\small
\begin{equation}\label{rate_1}
R = \log_2\left(1 + \frac{l_p G_t \gamma_t G_{\mathrm{r}}}{\sigma^2} \right),
\end{equation}}
\vspace{-2.5mm}\\
where $\gamma_t=\frac{P_t}{\sigma^2}$, with \eqref{rate_1} capturing the impact of beam steering accuracy, directional alignment, and environmental path loss on the end-to-end spectral efficiency. 
\vspace{-4mm}
\section{localization through VCSEL based LeRIS}\label{sec:III}
To determine the UE position within the PWE, received signal strength (RSS) measurements from VCSELs are used, since the distance-dependent attenuation of Gaussian beams carries geometric information for localization. However, the highly directionality of VCSELs improves spatial resolution but limits coverage, so detection and placement depend on whether the PD lies within each narrow beam. Specifically, as it can be seen in \eqref{eq:received_power}, the PD receives power only within the beam footprint, making it essential to ensure that enough VCSELs are deployed for reliable estimation of both the UE's position and orientation. This requirement can be reduced by operating each VCSEL in more than one optical mode. In practice, VCSELs can support more than one transverse lasing mode, and different modes lead to Gaussian beams with different waist radii and Rayleigh ranges, providing distinct distance dependences without increasing the number of emitters. The following proposition states that multi-mode emission enables unique localization is achieved with three VCSELs under appropriate geometric and identifiability conditions.

\begin{proposition}
The UE position $\mathbf{r}$ and orientation $\mathbf{n}$ can be uniquely determined from \emph{three} VCSELs if each emits sequentially in two distinguishable optical modes with known $\{P_{t,i}^{(m)}, w_{0,i}^{(m)}, z_{R,i}^{(m)}\}$ for $m\!\in\!\{a,b\}$, the VCSEL locations are non-collinear, and the unit vectors $\mathbf{u}_i=\frac{\mathbf{s}_i-\mathbf{r}}{\|\mathbf{s}_i-\mathbf{r}\|}$, $i=1,2,3$, are linearly independent at the solution.
\end{proposition}

\begin{IEEEproof}
Let the $i$-th VCSEL be at $\mathbf{s}_i\in\mathbb{R}^3$, and the UE at $\mathbf{r}\in\mathbb{R}^3$, with unit normal $\mathbf{n}\in\mathbb{R}^3$ satisfying $|\mathbf{n}|=1$. Owing to the high directionality of VCSELs, the irradiance can be considered aligned with the PD axis, i.e., $\phi=0^\circ$, while only measurements fulfilling $\mathbf{n}\cdot\mathbf{s}_i\geq 0$ are retained, corresponding to the FoV constraint of the PD. By assuming that the effect of noise is negligible, and setting each VCSEL to emit sequentially in two optical modes, indexed by $m\in\{a,b\}$, with known parameters consisting of the transmit power $P^{(m)}_{t,i}$, the waist radius $w^{(m)}_{0}$, and the Rayleigh range $z^{(m)}_{R}=\pi \left(w_{0}^{(m)}\right)^2/\lambda_o$, the received power at the UE for each mode is
\begin{equation}
P_i^{(m)}=\beta_i^{(m)}(d_i)\,(\mathbf{n}\cdot\mathbf{u}_i),
\label{eq:dual_power}
\end{equation}
where $\beta_i^{(m)}(d)=\frac{2A_{\mathrm{PD}}P^{(m)}_{t,i}}{\pi w_m^2(d)}$, and $w^{(m)}(d)=w_{0}^{(m)}\sqrt{1+\left(d/z_R^{(m)}\right)^2}$.

Since both modes originate from the same VCSEL, they share the same direction vector $\mathbf{u}_i$, thus the factor $(\mathbf{n}\cdot\mathbf{u}_i)$ in \eqref{eq:dual_power} is identical for $m=a$ and $m=b$. To eliminate this dependence on the unknown $\mathbf{n}$, we can consider the ratio $R_i$ of the received powers, which is defined as
\begin{equation}
R_i = \frac{P_i^{(a)}}{P_i^{(b)}} 
= \frac{\beta_i^{(a)}(d_i)}{\beta_i^{(b)}(d_i)}
= \frac{B_a}{B_b}\cdot \frac{1+\left(d_i/z^{(b)}_{R}\right)^2}{1+\left(d_i/z^{(a)}_{R}\right)^2},
\end{equation}
where $B_m=\frac{2A_{\mathrm{PD}}P_{t,i}^{(m)}}{\pi \left(w_{0}^{(m)}\right)^2}$. Thus, after some algebraic manipulations, the distance $d_i$ from the $i$-th VCSEL and the UE can be recovered in closed form as
\begin{equation}
d_i=\sqrt{\frac{1-\tilde{R}_i}{\frac{\tilde{R}_i}{\left(z^{(a)}_{R}\right)^2}-\frac{1}{\left(z^{(b)}_{R}\right)^2}}},
\label{d:RSS}
\end{equation}
where $\tilde{R}_i=\frac{R_i B_b}{B_a}$. With three such ranges from non-collinear VCSELs, the UE position $\mathbf{r}$ is obtained by trilateration as the common intersection of
\begin{equation}
\|\mathbf{r}-\mathbf{s}_i\|^2=d_i^2,\qquad i=1,2,3,
\label{eq:spheres}
\end{equation}
which reduces to a finite set. To avoid solving the full quadratic system, subtract the $i=1$ equation from $i=2,3$ to eliminate quadratic terms and obtain
\begin{equation}
2\begin{bmatrix}
(\mathbf{s}_2-\mathbf{s}_1)^\top\\
(\mathbf{s}_3-\mathbf{s}_1)^\top
\end{bmatrix}\mathbf{r}
=
\begin{bmatrix}
\|\mathbf{s}_2\|^2-\|\mathbf{s}_1\|^2-(d_2^2-d_1^2)\\[1mm]
\|\mathbf{s}_3\|^2-\|\mathbf{s}_1\|^2-(d_3^2-d_1^2)
\end{bmatrix},
\label{eq:lin_system}
\end{equation}
which describes the line of intersection of the two planes. Intersecting this line with any sphere in \eqref{eq:spheres} yields the unique feasible point inside the room.

With $\hat{\mathbf{r}}$ determined, we define
\begin{equation}
\hat{\mathbf{u}}_i=\frac{\mathbf{s}_i-\hat{\mathbf{r}}}{\|\mathbf{s}_i-\hat{\mathbf{r}}\|},\qquad
c_i^{(m)}=\frac{P_i^{(m)}}{\beta_i^{(m)}(d_i)}.
\end{equation}
Using a single mode (e.g., $m=a$), the orientation follows from
\begin{equation}
\underbrace{\begin{bmatrix}
\hat{\mathbf{u}}_1^\top\\ \hat{\mathbf{u}}_2^\top\\ \hat{\mathbf{u}}_3^\top
\end{bmatrix}}_{U\in\mathbb{R}^{3\times 3}}\mathbf{n}
=
\underbrace{\begin{bmatrix}
c_1^{(a)}\\ c_2^{(a)}\\ c_3^{(a)}
\end{bmatrix}}_{\mathbf{c}},
\label{eq:solve_n}
\end{equation}
which establishes a direct relation between the orientation vector and the normalized measurements once the position has been fixed. Thus, considering that $N_{\mathrm{V}}=3$, the system \eqref{eq:solve_n} reduces to a $3\times 3$ linear system, and, if $\det U\neq 0$, the orientation vector of the user can be written in closed form as
\begin{equation}
\tilde{\mathbf{n}}=U^{-1}\mathbf{c},
\end{equation}
which concludes the proof.
\end{IEEEproof}

\begin{remark}
Non-collinearity ensures a finite intersection in \eqref{eq:spheres} and linear independence of $\{\hat{\mathbf{u}}_i\}$ ensures $\det U\neq 0$, enabling unique recovery of position and orientation.
\end{remark}
\begin{remark}
Compared to LED-based LeRIS \cite{R2}, which typically needs at least four sources, three VCSELs suffice thanks to narrow beams and dual-mode operation.
\end{remark}
From Proposition~1, it follows that the UE position and orientation can be uniquely determined using three VCSELs when the effect of noise is considered negligible. However, in practice, the received optical signals at the PD are affected by noise, which perturbs the range estimates extracted from each VCSEL. Since these ranges form the basis of the localization procedure, it is essential to quantify the error incurred from each single VCSEL–UE link. In this direction, the following proposition provides the estimation error for the $i$-th VCSEL to the UE link.
\begin{proposition}
The estimation error for the $i$-th VCSEL to the UE link is given as
{\small
\begin{equation}\label{eq:vcsel_err_link}
\Delta d_{i}
= d_{i}\left[
1-\frac{1}{\sqrt{1+\alpha_{i}}}
\sqrt{\alpha_{i}-\frac{z_{R,i}^{2}}{d_{i}^{2}}}
\right],
\end{equation}}
where $\alpha_{i}= \tfrac{P_{\mathrm{LoS},i}}{P_{n,i}}$.
\end{proposition}
\begin{IEEEproof}
The localization error is defined as the difference between the actual and estimated distance, i.e.,
\begin{equation}
\Delta d_{i}=d_{i}-\hat d_{i}.
\label{eq:error_def}
\end{equation}
Considering that $\phi=0^\circ$ due to the high directivity of VCSELs, and that the beam lies within the UE’s FoV, the true distance $d_{i}$ can be written through \eqref{eq:received_power} as
{\small
\begin{equation}
d_{i}=\sqrt{z_{R,i}^{2}\!\left(\frac{2 P_{t,i} A_{\mathrm{PD}}\cos\psi_{i}}{\pi w_{0,i}^{2}P_{\mathrm{LoS},i}}-1\right)}.
\label{eq:d_PLOS}
\end{equation}}
\vspace{-2.5mm}\\
Furthermore, considering that $P_{r,i}=P_{\mathrm{LoS},i}\left(1+\frac{1}{\alpha_{i}}\right)$, by replacing $P_{\mathrm{LoS},i}$ with $P_{r,i}$ in \eqref{eq:d_PLOS}, we can obtain the estimated distance $\hat d_{i}$ as
{\small
\begin{equation}
\hat d_{i}=\sqrt{\frac{\alpha_{i}d_{i}^2-z_{R,i}^2}{\alpha_{i}+1}}.
\label{eq:d_Pest}
\end{equation}}
\vspace{-2.5mm}\\
Substituting \eqref{eq:d_PLOS} and \eqref{eq:d_Pest} into \eqref{eq:error_def} gives \eqref{eq:vcsel_err_link}, which concludes the proof.
\end{IEEEproof}\vspace{-2mm}   
\subsection{LeRIS-based Communication}\label{sec:LeRIS_Comm}\vspace{-1mm}    
After the UE position has been determined and the LeRIS–UE and LeRIS–LeRIS link visibility has been identified according to FoV and geometric reachability, the AP selects a feasible cascaded route toward the UE. Multiple candidate routes are evaluated, with feasibility captured by
\begin{equation}
\chi_{\pi}=\prod_{s\in S(\pi)} \chi_s,
\end{equation}
where $S(\pi)$ is the set of segments of route $\pi$, and $\chi_s\in\{0,1\}$ indicates whether segment $s$ satisfies the FoV and steering constraints. Specifically, $\chi_s$ is a binary indicator that equals $1$ only when the segment alignment falls within the PD's FoV as constrained in (4), and simultaneously satisfies the angular directivity requirement for viable steering defined in (16). As established in (19), the received signal over a cascaded LeRIS route is shaped jointly by the multiplicative channel gains of its segments and the corresponding path losses, thus the optimal route maximizes the spectral efficiency in (20). So, the end-to-end spectral efficiency for the UE can be expressed as
\begin{equation}
R=\max_{\pi} R_{\pi},
\end{equation}
with
\begin{equation}
R_{\pi}=\chi_{\pi}\,
\log_2\!\left(
1+\frac{G_t\,l_p(\pi)\,G_{\mathrm{r}}(\pi)}{\sigma^2}
\right).
\label{eq:rate_singleUE}
\end{equation}
Here, $l_p(\pi)$ denotes the cumulative path loss along the segments of $\pi$, and $G_{\mathrm{r}}(\pi)$ is the cascaded beamforming gain from the participating LeRIS panels configured hop by hop toward the next LeRIS and the receive gain of the directional UE antenna within the UE FoV. The steering angles for each LeRIS hop and for the final LeRIS–UE segment follow from the estimated UE position as in~\cite{R2}, completing the end-to-end configuration for the selected route.
 
\section{Simulation Results}\label{sec:IV}\vspace{-1mm}   
In this section, we evaluate the performance of an indoor system with a single UE and four LeRIS panels, whose corresponding parameter values for the VCSELs are listed in Table~\ref{table1}, while the communication parameters are provided in Table~\ref{table2}. Each LeRIS is equipped with 24 VCSELs placed along its perimeter, arranged so that their beams jointly span an azimuthal sector of $120^\circ$ with an elevation span of $60^\circ$, while each VCSEL rotates to cover a $5^\circ$ azimuthal segment and thereby provide both wide angular coverage and fine resolution. As a result, the deployed beams complement one another to span the served area, enabling user localization and supporting LeRIS-based communication. Moreover, the UE coordinates are uniformly sampled relative to the corner origin within $x_u \in [0,10]$~m, $y_u \in [0,10]$~m, and $z_u=1.5$~m, corresponding to a typical indoor user height aligned with the LeRIS plane. The azimuth orientation angle $\phi_{\mathrm{UE}}$ of the user is uniformly distributed over $[0,2\pi]$, while the elevation orientation is fixed at $\theta_{\mathrm{UE}}=0$, reflecting alignment with the LeRISs’ centres. Finally, to quantify the system performance under realistic conditions, we conduct a Monte Carlo simulation with $10^5$ iterations, where in each iteration, both the user position and orientation are randomly generated within the specified bounds.

\begin{table}[h]
\vspace{-2mm}
	\renewcommand{\arraystretch}{1.1}
    \captionsetup{font=small}
	\caption{\textsc{Simulation Parameters for VCSELs}}
	\label{PWE_char}
	\centering
	\begin{tabular}{ll}
		\hline
		\bfseries Parameter & \bfseries Value \\
		\hline\hline
        LeRIS	1 Centre coordinates	 	    & $\left(0, 5, 1.5\right)$ 	
        \\
		LeRIS	2 Centre coordinates 	    & $\left(10, 5, 1.5\right)$ 
        \\
  	    LeRIS	3 Centre coordinates 	    & $\left(5, 0, 1.5\right)$ 	
        \\
        LeRIS	4 Centre coordinates	 	    & $\left(5, 10, 1.5\right)$ 
        \\
        PD area 	    & $A_\mathrm{PD} = 1$ cm${}^2$
        \\
        Transmit power 	    & $P_{t,i} = 10$ mW 	
        \\
        Beam waist          & $\omega_0 = 5.6$ $\mu$m
        \\
        Speed of light & $c = 3\times10^8 \mathrm{m/s}$
        \\
        VCSEL wavelength  & $\lambda_o = 950$ nm
        \\
         Optical bandwidth  & $B_o = 1$ GHz
        \\
         Relative intensity noise & $\mathrm{RIN} = -155$ dB/HZ
        \\
        Preamplifier noise figure & $F_{n}$ = 5 dB
        \\
        PD FoV           & $\Psi = 90\degree$
        \\
        Noise variance	 	    & $P_{n,i} = 2.5\times 10^{-12}$A$^2$ 
        \\
        Load resistance 	    & $R_\mathrm{L} = 50\Omega$ 
        \\
         Absolute temperature & $T = 300 \mathrm{K}$
         \\
         Boltzmann constant & $k_B =1.380649 \times10^{-23} \mathrm{K/J}$
         \\
        PD responsivity 	    & $R_{\mathrm{PD}} = 0.7$ A/W
        \\
		\hline
	\end{tabular}
    \label{table1}
\end{table}
\vspace{-3mm} 
\begin{table}
\vspace{-2mm}
	\renewcommand{\arraystretch}{1}
    \captionsetup{font=small}
	\caption{\textsc{Simulation Parameters for mmWave}}
	\label{PWE_char}
	\centering
	\begin{tabular}{ll}
		\hline
		\bfseries Parameter & \bfseries Value \\
		\hline\hline
        mm-wave Wavelength	    & $ \lambda_m= 10^{-2}  $m
        \\
        Transmit power & $P_t=1$ W
        \\
        Transmitter gain & $G_t=10$ dB
        \\
         Angle of directivity & $\theta_r= \frac{\pi}{3}$
         \\
        Reference distance & $d_{r,0}=1$ m
        \\
        Path loss exponent & $n_i=2$
        \\
        Element spacing & $D=\frac{\lambda_m}{2}$
        \\
        RIS efficiency & $n_\mathrm{eff}=100\%$
        \\
        AWGN variance & $\sigma^2 = -130$ dB 
        \\
		\hline
	\end{tabular}
    \label{table2}
\end{table}

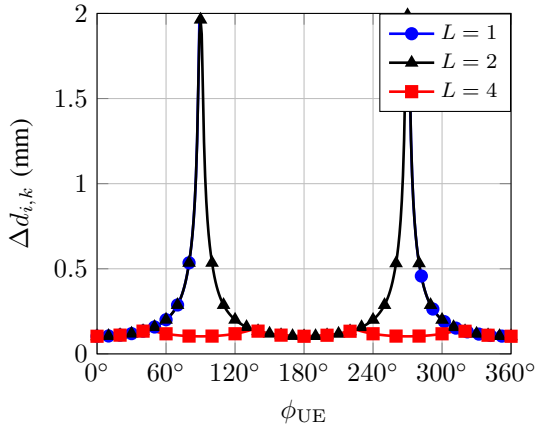
\begin{figure}[ht]
\vspace{-0.3cm}
  \centering
  \begin{tikzpicture}
    \begin{axis}[
      width=0.8\linewidth,
      xlabel = {$\phi_{\mathrm{UE}}$},
      ylabel = {$\Delta d_{i,k}$ (mm)},
      xmin = 0, xmax = 360,
      ymin = 0, ymax = 2.0,                 
      ytick = {0,0.5,1,1.5,2},
      xtick = {0,60,120,180,240,300,360},
      xticklabels = {$0^\circ$, $60^\circ$, $120^\circ$, $180^\circ$, $240^\circ$, $300^\circ$, $360^\circ$},
      grid = major,
      restrict y to domain=0:2.0,           
      scaled y ticks = false,               
      yticklabel style = {/pgf/number format/fixed},
      legend entries = {{$L=1$}, {$L=2$}, {$L=4$}},
      legend cell align = {left},
      legend style={font=\footnotesize, at={(1,1)}, anchor=north east}
    ]
      \addplot[blue,  mark=*,         mark repeat=10, mark size=2, line width=1pt]
        table[y expr=1000*\thisrowno{1}] {Fig1/Fig1_single.txt};
      \addplot[black, mark=triangle*, mark repeat=10, mark size=2, line width=1pt]
        table[y expr=1000*\thisrowno{1}] {Fig1/Fig1_2RIS.txt};
      \addplot[red,   mark=square*,   mark repeat=20, mark size=2, line width=1pt]
        table[y expr=1000*\thisrowno{1}] {Fig1/Fig1_optimal.txt};
    \end{axis}
  \end{tikzpicture}
  \vspace{-2mm}
  \captionsetup{font=small}
  \caption{$\Delta d$ versus azimuth angle for different values of $L$.}
  \vspace{-2mm}
  \label{fig:phi_vs_error}
\end{figure}
Fig.~\ref{fig:phi_vs_error} shows the localization error $\Delta d_{i}$ versus the azimuth angle for different numbers of active panels. With a single panel at $0^\circ$, the error remains low within its $60^\circ$ footprint but increases sharply outside this range, reflecting the high precision but limited coverage of directional VCSEL beams. Activating two opposite LeRIS panels extends coverage and lowers error over a wider range, while full activation of all four panels maintains consistently low error across the full $360^\circ$ orientation range. This demonstrates that cooperative VCSEL-based LeRIS panels can ensure robust and accurate localization regardless of user orientation. Compared to wide-beam LED systems, the use of narrow-beam VCSELs with rotational scanning provides significantly higher accuracy, validating the efficiency of the proposed architecture.

\begin{figure}[ht]
\vspace{-0.3cm}
\centering
\begin{tikzpicture}
    \begin{axis}[
        width=0.8\linewidth,
        xlabel = {SNR(dB)},
        ylabel = { $R$ (bps/Hz)},
        ymin = 0,
        ymax = 16,
        xmin = 90,
        xmax = 130,
        xtick = {90,95,...,130},
        ytick = {0,4,...,16},
        grid = major,
        legend cell align = {left},
                legend style={font=\footnotesize},
                legend style={at={(0,1)},anchor=north west},
    ]

    \addplot[
        mark=triangle*,
        color=blue,
        mark repeat = 2,
                mark size = 2,
                line width = 1pt,
                style = solid,
    ]
    table{FigB_ICC/1RIS.txt};
    \addlegendentry{$L=1$}

    \addplot[
        mark=*, 
        color=black,
        mark repeat = 2,
                mark size = 2,
                line width = 1pt,
                style = solid,
    ]
    table{FigB_ICC/2RIS_fig2_5k.txt};
    \addlegendentry{$L=2$}

    \addplot[
        mark=square*,
        color=red,
        mark repeat = 2,
                mark size = 2,
                line width = 1pt,
                style = solid,
    ]
    table{FigB_ICC/4RIS.txt};
    \addlegendentry{$L=4$}
    
    \end{axis}
\end{tikzpicture}
        \vspace{-2mm}
        \captionsetup{font=small}
\caption{$R$ versus SNR for various numbers of LeRIS panels.}
    \label{fig:figure3}
     \vspace{-0.4cm}
\end{figure}
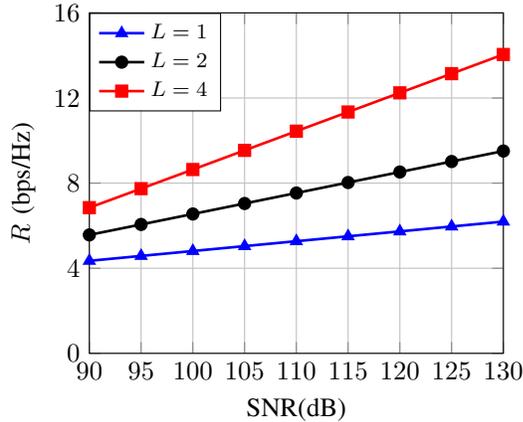
Fig.~\ref{fig:figure3} presents the spectral efficiency $R$ as a function of SNR for $L = 1$, $L = 2$, and $L = 4$ LeRIS panels, each with $N = 2500$ elements. As expected, increasing the number of active panels enhances spectral efficiency by offering more reflective paths and improving localization precision. This enables reliable high-rate communication across a broad SNR range, especially at higher SNRs where reflected beams are more effectively exploited.

\begin{figure}[ht]
\vspace{-0.3cm}
\centering
\begin{tikzpicture}
    \begin{axis}[
        width=0.8\linewidth,
        xlabel = {$N$},
        ylabel = { $R$ (bps/Hz)},
        ymin = 0,
        ymax = 16,
        xmin = 100,
        xmax = 2500,
        xtick = {100,500,1000,1500,2000,2500}, 
        ytick = {0,4,...,16},
        grid = major,
        legend cell align = {left},
                legend style={font=\footnotesize},
                legend style={at={(0,1)},anchor=north west},
    ]

    \addplot[
        mark=triangle*,
        color=blue,
        mark repeat = 5,
                mark size = 2,
                line width = 1pt,
                style = solid,
    ]
    table{FigC_ICC/Single_RIS_data1.txt};
    \addlegendentry{$L=1$}

    \addplot[
        mark=*, 
        color=black,
        mark repeat = 5,
                mark size = 2,
                line width = 1pt,
                style = solid,
    ]
    table{FigC_ICC/figure_Two_RIS_data1.txt};
    \addlegendentry{$L=2$}

    \addplot[
        mark=square*,
        color=red,
        mark repeat = 5,
                mark size = 2,
                line width = 1pt,
                style = solid,
    ]
    table{FigC_ICC/Four_RIS_data1.txt};
    \addlegendentry{$L=4$}
    
    \end{axis}
\end{tikzpicture}
        \vspace{-2mm}
        \captionsetup{font=small}
\caption{$R$ versus $N$ for various numbers of LeRIS panels}
    \label{fig:figure4}
     \vspace{-0.4cm}
\end{figure}
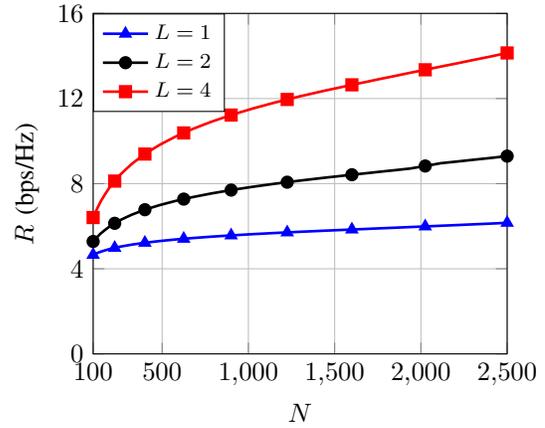
Fig.~\ref{fig:figure4} illustrates the spectral efficiency $R$ as a function of the number of reflecting elements $N$ per LeRIS panel for three deployment configurations. When only a single panel is active, the system exhibits restricted angular coverage, which limits the likelihood of effective beam alignment and reduces the spatial diversity available for both localization and communication. This limitation leads to degraded spectral efficiency, particularly in scenarios requiring omnidirectional coverage or high-reliability links. As additional LeRIS panels are activated, the PWE gains greater spatial flexibility, allowing it to capture stronger reflection paths and provide broader angular support for the user. The resulting improvements in beamforming gain and localization precision translate into a marked increase in spectral efficiency. Notably, the $L = 4$ configuration consistently outperforms the $L = 1$ case across all values of $N$, achieving more than twice the throughput in several regimes. These results underscore the critical role of both panel multiplicity and element scalability in enhancing the overall performance of VCSEL-based LeRIS systems.\vspace{-4mm}   
\section{Conclusion}
This paper presented a VCSEL-based LeRIS architecture for PWEs that enables joint user localization, orientation estimation, and mmWave communication. By exploiting the Gaussian beam profile and dual-mode diversity of VCSELs, closed-form solutions were derived to estimate user position and orientation using only three VCSEL sources integrated along the panel perimeter. A distance estimation error expression was derived to evaluate localization accuracy. Finally, simulation results confirmed that the proposed system achieves millimeter-level localization accuracy and robust orientation estimation while maintaining high spectral efficiency. These findings demonstrate that VCSEL-integrated LeRIS panels offer a practical and scalable solution for enabling multi-functional RIS operation in future 6G wireless networks.

\vspace{-3mm}   


\end{document}